\newif\ifFull
\Fullfalse

\documentclass{article}
\usepackage{graphicx}
\usepackage{verbatim}
\usepackage{url}
\newcommand{\R}{{\bf R}}
\newcommand{\cP}{{\mathcal{P}}}

\newcommand{\define}[1]{{\bfseries\itshape{#1}}}
\newtheorem{theorem}{Theorem}[section]
\newtheorem{lemma}[theorem]{Lemma}
\newenvironment{proof}{\noindent{\it Proof:}\hspace*{1em}}{\medskip}
\newcommand{\qed}{\rule{6pt}{6pt}}
\begin{document}
\title{Planar Drawings of Higher-Genus Graphs}

\iffalse
\author{
Christian A.~Duncan\inst{1}
\and Michael T.~Goodrich\inst{2} 
\and Stephen G.~Kobourov\inst{3}
}

\institute{
    Dept.~of Computer Science,
    Louisiana Tech Univ.\\
    {\tt \url{http://www.latech.edu/~duncan/} }
\and
Dept.~of Computer Science,
Univ.~of California, Irvine\\
{\tt \url{http://www.ics.uci.edu/~goodrich/}}\\
\and
    AT\&T Research Labs,
    Florham Park, NJ\\
    {\tt \url{http://www.research.att.com/info/skobourov}}
}
\else
\author {
\sc Christian A.~Duncan
\\\normalsize Department of Computer Science
\\\normalsize Louisiana Tech University
\\\normalsize {\tt \url{http://www.latech.edu/~duncan/} }
\and
\sc Michael T. Goodrich
\\\normalsize Department of Computer Science
\\\normalsize University of California, Irvine
\\\normalsize {\tt \url{http://www.ics.uci.edu/~goodrich/}}
\and 
\sc Stephen G. Kobourov
\\\normalsize Department of Computer Science
\\\normalsize University of Arizona
\\\normalsize {\tt \url{http://www.cs.arizona/~kobourov/}}
}
\fi

\maketitle

\begin{abstract}
In this paper, we give polynomial-time
algorithms that can take a graph $G$ with a given combinatorial 
embedding on an orientable surface $\cal S$ of genus $g$ 
and produce a planar drawing of $G$ in $\R^2$, 
with a bounding face defined by a polygonal
schema $\cal P$ for $\cal S$.
Our drawings are planar, but they allow for multiple copies of vertices and
edges on $\cal P$'s boundary, which is a common way of visualizing
higher-genus graphs in the plane.
Our drawings can be defined with respect to either a canonical
polygonal schema or a polygonal cutset schema, which provides an
interesting tradeoff, since canonical schemas have fewer sides,
and have a nice topological structure, but they can have many more repeated
vertices and edges than general polygonal cutsets.
As a side note, we show that it is NP-complete to
determine whether a given graph embedded in a genus-$g$ surface has a
set of $2g$ fundamental cycles with vertex-disjoint interiors, which
would be desirable from a graph-drawing perspective.
\end{abstract}

\section{Introduction}
The \emph{classic} way of drawing a graph $G=(V,E)$ in $\R^2$ involves
associating each vertex $v$ in $V$ with a unique point $(x_v,y_v)$ and associating with each edge $(v,w)\in E$ an open Jordan
curve that has  $(x_v,y_v)$ and $(x_w,y_w)$ as its endpoints.
If the curves associated with the edges in a classic drawing of 
$G$ intersect only at
their endpoints, then (the embedding of) $G$ is a \emph{plane graph}.
Graphs that admit plane graph representations are
\emph{planar graphs}, and
there has been a voluminous amount of work on algorithms on classic drawings
of planar graphs.
Most notably, planar graphs can be drawn with vertices assigned to
integer coordinates in an $O(n)\times O(n)$ grid, which is often a 
desired type of classic drawing known as a \emph{grid drawing}.
Moreover, there are planar graph drawings that use only straight line
segments for edges~\cite{fpp-hdpgg-90}.

The beauty of plane graph drawings is that, by avoiding edge
crossings, confusion and clutter in the drawing is minimized. Likewise,
straight-line drawings further improve graph visualization by 
allowing the eye to easily follow connections between adjacent vertices.
In addition, grid drawings enforce a natural separation between
vertices, which further improves readability.
Thus, a ``gold standard'' in classic drawings is to produce planar
straight-line grid drawings and, when that is not easily done, to
produce planar grid drawings with edges drawn as simple polygonal chains.

Unfortunately, not all graphs are planar. 
So drawing them in the classic way requires some compromise in the
gold standard for plane drawings.
In particular, any classic drawing of a non-planar graph must necessarily 
have edge crossings, and minimizing the number of crossings is NP-hard~\cite{gj-cninc-83}. One point of hope for improved drawings of non-planar graphs is to
draw them crossing-free on surfaces of higher genus, such as 
toruses, double toruses, or, in general, a surface topologically
equivalent to a sphere with $g$ handles, that is, a \emph{genus-$g$}
surface.
Such drawings are called 
\emph{cellular} embeddings or \emph{$2$-cell} embeddings, since they
partition the genus-$g$ surface into a collection of cells that are
topologically equivalent to disks. 
As in classic drawings of planar graphs,
these cells are called \emph{faces}, and it is easy to see that such
a drawing would avoid edge crossings. 

In a fashion analogous to the case with planar graphs,
cellular embeddings of graphs in a genus-$g$ surface can be
characterized combinatorially.
In particular, it is enough if we just have a rotational order of the edges
incident on each vertex in a graph $G$ to determine a 
combinatorial embedding of $G$ on a surface (which has that ordering
of associated curves listed counterclockwise around each vertex).
Such a set of orderings is called a \emph{rotation system} and, since
it gives us a combinatorial description of the set of faces, $F$,
in the embedding,
it gives us a way to determine the genus of the (orientable)
surface that $G$ is embedded into by using the \emph{Euler characteristic},
$|V|-|E|+|F| = 2- 2g,$
which also implies that $|E|$ is $O(|V|+g)$~\cite{mt-gos-01}.

Unfortunately, given a graph $G$, 
it is NP-hard to find
the smallest $g$ such that $G$ has a combinatorial cellular embedding on a
genus-$g$ surface~\cite{t-ggpnp-89}.
This challenge need not be a deal-breaker in practice, however, for there
are heuristic algorithms for producing such combinatorial
embeddings (that is, consistent rotation systems)~\cite{Chen1997317}.
Moreover, higher-genus graphs often come together with 
combinatorial embeddings in practice, as in many computer graphics
and mesh generation applications.

In this paper, we assume that we are given a
combinatorial embedding of a graph $G$ on a genus-$g$ surface, $\cal S$, and are
asked to produce a geometric drawing of $G$ that respects
the given rotation system. Motivated by the gold standard for planar graph drawing and by the fact that computer screens and physical printouts are still primarily two-dimensional display surfaces, the approach we take is to draw $G$ in the plane rather than on some embedding of $\cal S$ in $\R^3$.

Making this choice of drawing paradigm, of course, requires that we
``cut up'' the genus-$g$ surface, $\cal S$,
and ``unfold'' it so that the resulting sheet is topologically
equivalent to a disk.
The traditional method for
performing such a cutting is with a \emph{canonical polygonal schema}, 
$\cal P$, which
is a set of $2g$ cycles on $\cal S$ all containing a common point, $p$,
such that cutting $\cal S$ along these cycles results in a
topological disk.
These cycles are \emph{fundamental} in that each
of them is a continuous closed curve on $\cal S$ that cannot be
retracted continuously to a point.
Moreover,
these fundamental 
cycles can be paired up into complementary sets of cycles, $(a_i,b_i)$,
one for each handle, so that if we orient the sides of $\cal P$, then a
counterclockwise ordering of the sides of $\cal P$ can be listed as
$
a_1b_1a_1^{-1}b_1^{-1}
a_2b_2a_2^{-1}b_2^{-1}
\,
\ldots
\,
a_{g}b_{g}a_{g}^{-1}b_{g}^{-1},
$
where $a_i^{-1}$ ($b_i^{-1})$ is 
a reversely-oriented copy of $a_i$, so that these two sides
of $\cal P$
are matched in orientation on $\cal S$.
Thus, the canonical polygonal schema for a genus-$g$ surface $\cal S$
has $4g$ sides that are pairwise identified.


Because we are interested in drawing the graph $G$ and 
not just the topology of $\cal S$, 
it would be preferable if the fundamental cycles are also cycles in
$G$ in the graph-theoretical sense.
It would be ideal if these cycles form a canonical polygonal
schema with no repeated vertices other than the common one.
This is not always possible~\cite{lpvv-ccpso-01} and furthermore, as we show
in~\cite{XXX}, the problem of finding a set of $2g$ fundamental
cycles with vertex-disjoint interiors in a combinatorially 
embedded genus-$g$ graph is NP-complete.
There are two natural choices, both of which we explore in this
paper:
\begin{itemize}
\item Draw $G$ in a polygon $P$ corresponding to 
a canonical polygonal schema, $\cal P$,
possibly with repeated vertices and edges on its boundary.
\item Draw $G$ in a polygon $P$ corresponding to 
a polygonal schema, $\cal P$, that is not canonical.
\end{itemize}
In either case, the edges and vertices on the boundary
of $P$ are repeated (since we ``cut'' $\cal S$ along these edges and
vertices). Thus, we need labels in our drawing of $G$ to identify the
correspondences. 
Such planar drawings of $G$ inside a polygonal schema $\cal P$ 
are called \emph{polygonal-schema} drawings of $G$.
There are three natural aesthetic criteria such drawings should satisfy:
\begin{enumerate}
\item
\emph{Straight-line edges:} 
All the edges in a polygonal-schema drawing should be rendered as
polygonal chains, or straight-line edges, when possible.
\item
\emph{Straight frame:}
Each edge of the polygonal schema 
should be rendered as a straight line segment, with the
vertices and edges of the corresponding fundamental cycle, placed along this segment.
We refer to such a polygonal-schema drawing as having 
a \emph{straight frame}.
\item
\emph{Polynomial area:}
Drawings should have polynomial area when they are
normalized to an integer grid.
\end{enumerate}

It is also possible to avoid repeated vertices and instead use a classic 
graph drawing paradigm, by transforming
the fundamental polygon rendering using
polygonal-chain edges that run through
``overpasses'' and ``underpasses'' as in road networks, so as to illustrate
the topological structure of $G$; see Fig.~\ref{fig-fund}.

\begin{figure}[tbh]
\begin{center}
\includegraphics[angle=-90,width=12cm]{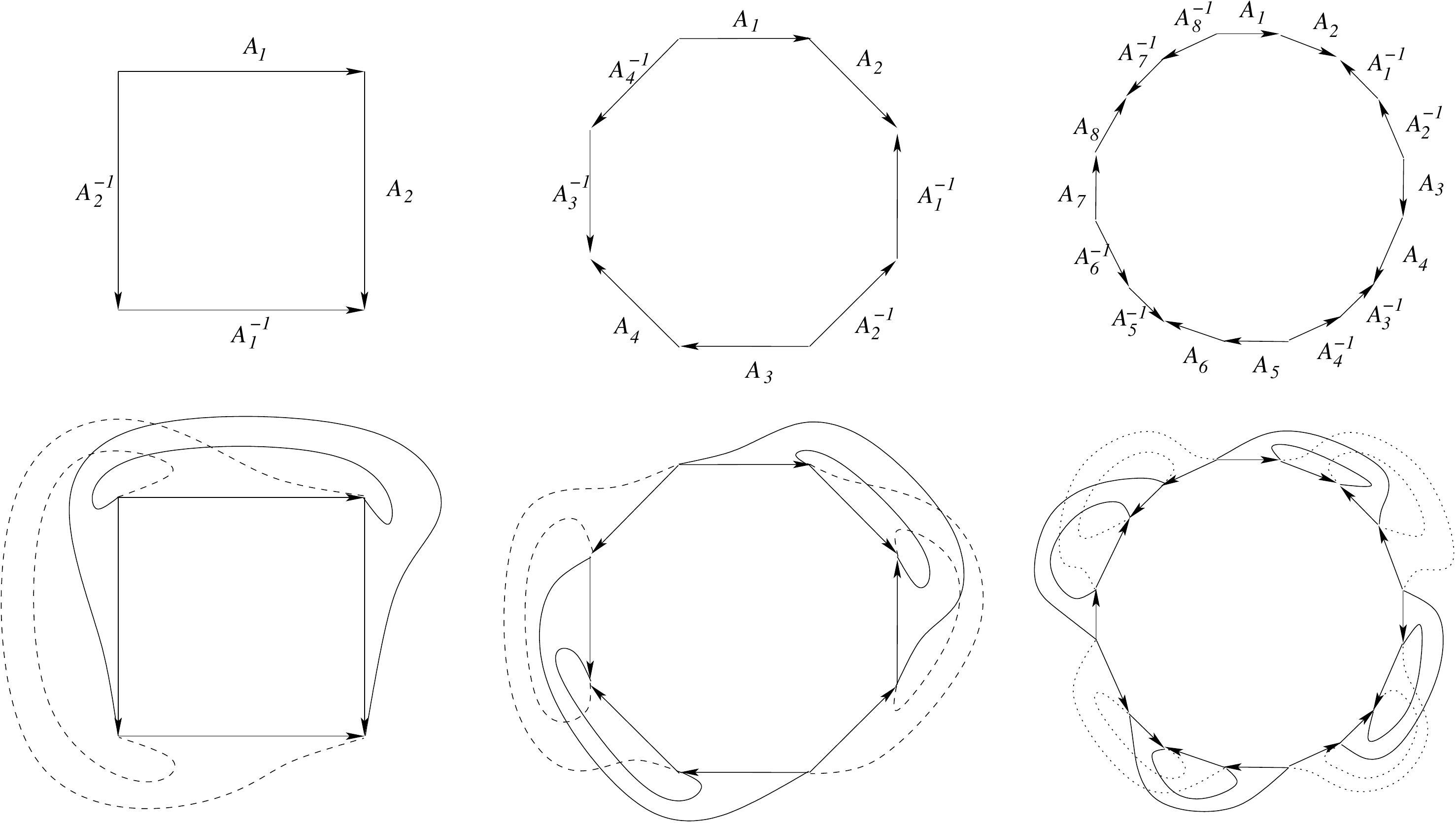}
\end{center}
\caption{\small\sf First row: Canonical polygonal schemas
for graphs of genus one, two and four. Second row: Unrolling the high
genus graphs with the aid of the overpasses and underpasses.}
\label{fig-fund}
\end{figure}


\paragraph{Our Contributions.}
We provide several methods for producing
planar polygonal-schema drawings of higher-genus graphs.
In particular, we provide four algorithms, one for torodial ($g=1$)
graphs and three for non-toroidal ($g>1$) graphs.
Our algorithm for toroidal 
graphs simultaneously achieves the three
aesthetic criteria for polygonal schema drawings: it uses
straight-line edges, a straight frame, and polynomial area.
The three algorithms for non-toroidal
graphs, \emph{Peel-and-Bend},
\emph{Peel-and-Stretch}, and \emph{Peel-and-Place},
achieve two of the tree aesthetic criteria, and differ in which criteria they fail to meet.

\section{Finding Polygonal Schemas}
Suppose we are given a graph $G$ together with its cellular embedding
in a genus-$g$ surface, $\cal S$.
An important first step in all of our algorithms involves our finding a
polygonal schema, $\cal P$,  for $G$, that is, a set of cycles
in $G$ such that cutting $\cal S$ along these cycles results in a
topological disk.
We refer to this as the \emph{Peel}
step, since it involves cutting the surface $\cal S$ until it becomes
topologically equivalent to a disk.
Since these cycles form the sides of the fundamental polygon we will
be using as the outer face in our drawing of $G$, it is desirable
that these cycles be as ``nice'' as possible with respect to drawing
aesthetics. 

\subsection{Trade-offs for Finding Polygonal Schemas}
Unfortunately, some desirable properties are not effectively
achievable.
As Lazarus \textit{et al.}~\cite{lpvv-ccpso-01} show, it is not always possible
to have a canonical polygonal schema $\cal P$ such that each
fundamental cycle in $\cal P$ has a distinct
set of vertices in its interior (recall that the interior of a
fundamental cycle is the set of vertices distinct from the common
vertex shared with its complementary fundamental cycle---with this
vertex forming a corner of a polygonal schema).
In addition, we show in~\cite{XXX} that finding a vertex-disjoint
set of fundamental cycles is NP-complete.
So, from a practical point of view, we have two choices with respect
to methods for finding polygonal schemas.

\paragraph{Finding a Canonical Polygonal Schema.}
As mentioned above, a canonical polygon schema of a graph $G$ 2-cell
embedded in a surface of genus $g$ consists of 
$4g$ sides, which correspond to $2g$ fundamental cycles all
containing a common vertex.
Lazarus {\it et al.}~\cite{lpvv-ccpso-01} show that one can find such a schema
for $G$ in $O(gn)$ time and with total size $O(gn)$, and they
show that this bound is within a constant factor of optimal in the worst case,
where $n$ is the total combinatorial complexity of $G$ (vertices, edges, and faces), which is $O(|V|+g)$.

\paragraph{Minimizing the Number of Boundary Vertices in a Polygonal Schema.}
Another optimization would be to minimize the number of vertices in the
boundary of a polygonal schema.
Erikson and Har-Peled~\cite{eh-ocsid-02} show that this problem is NP-hard,
but they provide an $O(\log^2 g)$-approximation algorithm that runs
in $O(g^2n \log n)$ time, and they give an exact algorithm that
runs in $O(n^{O(g)})$ time.

In our \emph{Peel} step, we assume that we use one of these two
optimization criteria to find a polygonal schema, which either
optimizes its number of sides to be $4g$, as in the canonical case,
or optimizes the number of vertices on its boundary, which will be
$O(gn)$ in the worst case either way.
Nevertheless, for the sake of concreteness, we often describe our
algorithms assuming we are given a canonical polygonal schema.
It is straight-forward to adapt these algorithms for non-canonical schemas. 
\subsection{Constructing Chord-Free Polygonal Schemas}
\label{sec:chordFree}
In all of our algorithms the first step, {\em Peel}, constructs a polygonal schema of the input graph $G$.
In fact, we need a polygonal schema,  $\cal P$, in which
there is no chord connecting two vertices on the same side of $\cal P$.
Here we show how to transform any polygonal
schema into a chord-free polygonal schema.

In the {\em Peel} step, we cut the graph $G$ along a canonical set of 
$2g$ fundamental cycles getting two
copies of the cycle in $G^*$, the resulting planar graph. 
For each of the two pairs of every fundamental cycle there may be chords.  
If the chord connects two vertices that are in different copies of the cycle in $G^*$ then this
is a chord that {\em can} be drawn with a straight-line edge and hence
does not create a problem. 
However, if the chord connects two vertices in the
same copy of the cycle in $G^*$, then we will not be able to place all
the vertices of that cycle on a straight-line segment; see Figure~\ref{fig:chordFree}(a). 
We show next that a new chord-free polygonal schema can be efficiently determined from the original schema.

\begin{theorem}
\label{theorem:chordFree}
Given a graph $G$ combinatorially embedded in a genus-$g$ surface 
and a canonical polygonal schema $\cP$ on $G$ with a common vertex $p$,
a \emph{chord-free polygonal schema} $\cP^*$ can be found in $O(gn)$ time.
\end{theorem}

\begin{proof}
We first use the polygonal schema to cut the embedding of $G$ into a
topological disk; see Fig.~\ref{fig:chordFree}(a).
Notice this cutting will cause certain vertices to be split into multiple
vertices.
For each fundamental cycle in $c_i \in \cP$, we stitch the disk graph back together
along this cycle forming a topological cylinder.
The outer edges (left and right) of the cylinder along this stitch will have two copies
of the vertex $p$, say $p_1$ and $p_2$.
We perform a shortest path search from $p_1$ to $p_2$.
This path becomes our new fundamental cycle $c^*_i$, (since $p_1$ and $p_2$ are 
the same vertex in $G$).
Observe that this cycle must be chord-free or else the path chosen was not the
shortest path; see Fig.~\ref{fig:chordFree}(b).
We then cut the cylinder along $c^*_i$ and proceed to $c_{i+1}$.
The resulting set, $\cP^*=\{c^*_1, c^*_2, \dots, c^*_{2g}\}$, is 
therefore a collection of chord-free fundamental cycles
all sharing the common vertex $p$.
\qed
\end{proof}

It should be noted that, although each cycle $c^*_i$ is at the time of its creation 
a shortest path from the two copies of $p$, 
these cycles are {\em not} the shortest fundamental cycles possible.
For example, a change in the cycle of $c_{i+1}$ could introduce a shorter possible path 
for $c^*_i$, but not additional chords.

\begin{figure}
\begin{center}
\includegraphics[width=6cm]{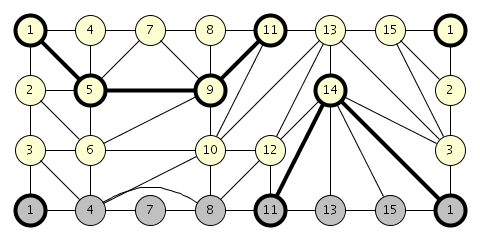}
\includegraphics[width=6cm]{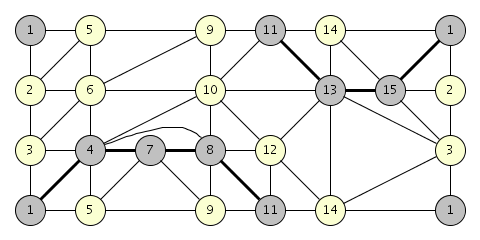}\\
(a) \hspace{6cm} (b) 
\end{center}
\caption{\small\sf 
(a) A graph embedded on the torus that has been cut into a topological disk
using the cycles $1,2,3$ and $1,4,7,8,11,13,15$ with chord $(4,8)$.
The grey nodes correspond to the identical vertices above.
The highlighted path represents a shortest path between the two copies of vertex 1.
(b) The topological disk {\em after} cutting along this new fundamental cycle.
The grey nodes show the old fundamental cycle.
}
\label{fig:chordFree}
\end{figure}

\section{Straight Frame and Polynomial Area}

In this section, we describe our algorithms that construct a drawing
of $G$ in a straight frame using polynomial area.
Here we are given an embedded 
genus-$g$ graph $G=(V,E)$ along with a chord-free polygonal schema, $\cal P$,
for $G$ from the {\em Peel} step. 
We rely on a modified version of the 
algorithm of de Fraysseix, Pach and Pollack~\cite{fpp-hdpgg-90} for the drawing. 
Sections~\ref{sec:embed} and~\ref{sec:embedHigher} describe the details for $g=1$ and for $g>1$,
respectively. 
In the latter case we introduce up to $O(k)$ edges with single bends where $k$ is the number
of vertices on the fundamental cycles.
Thus, we refer to the algorithm for non-toroidal graphs as the
\emph{Peel-and-Bend} algorithm.

\subsection{Grid Embedding of Toroidal Graphs}
\label{sec:embed}

For toroidal graphs we are able to achieve all three aesthetic criteria: straight-line edges, straight frame, and polynomial area. 

\begin{theorem}
\label{theorem:gOneEmbed}
Let $G^*$ be an embedded planar graph and
$\cP=\{P_1, P_2, \dots, P_{4g}\}$ in $G^*$ be a collection of $4g$ paths such that
each path $P_i=\{p_{i,1}, p_{i,2}, \dots, p_{i,k_i}\}$ is chord-free,
the last vertex of each path matches the
first vertex of the next path,
and when treated as a single cycle, $\cP$ forms the external face of $G^*$.
If $g=1$, we can in linear time draw $G^*$ on an $O(n) \times O(n^2)$ grid 
with straight-line
edges and no crossings in such a way that, 
for each path $P_i$ on the external face, 
the vertices on that path form a straight line.
\end{theorem}


\begin{proof}
For simplicity, we assume that every face is a triangle, 
except for the outer face (extra edges can be added and later removed).
The 
algorithm of de Fraysseix, Pach and Pollack (dPP)~\cite{fpp-hdpgg-90} does not directly solve our problem because of the additional requirement for the drawing of the external face.
In the case of $g=1$, the additional requirement is that the graph must be drawn so that the external face forms a rectangle,
with $P_1$ and $P_3$ as the top and bottom horizontal boundaries and $P_2$ and $P_4$ as the right and left boundaries.

Recall that the dPP algorithm computes a canonical labeling of the vertices of the 
input graph and inserts them one at a time in that order while ensuring that 
when a new vertex is introduced it can ``see'' all of its already inserted neighbors. 
One technical difficulty lies in the proper placement of the top row of vertices.
Due to the nature of the canonical order, we cannot force the top row of vertices
to all be the last set of vertices inserted, unlike the bottom row which can be the first set inserted.
Consequently, we propose an approach similar to that of Miura, Nakano, and Nishizeki~\cite{nakano4connected01}.  First, we 
split the graph into two parts (not necessarily of equal size), perform a modified embedding on both pieces,
invert one of the two pieces, and stitch the two pieces together.

\begin{lemma}
\label{lemma:splitA}
Given an embedded plane graph $G$ that is fully triangulated except for the external face
and two edges $e_l$ and $e_r$ on that external face,
it is possible in linear time to partition $V(G)$ into two subsets $V_1$ and $V_2$ such that
\begin{enumerate}
\item the subgraphs of $G$ induced by $V_1$ and $V_2$, called $G_1$ and $G_2$, are both
connected subgraphs;
\item for edges $e_l=(u_l,v_l)$ and $e_r=(u_r,v_r)$, we have $u_l,u_r \in V_1$ and $v_l,v_r \in
  V_2$; 
\item the union $U$ of the set of faces in $G$ that are not in $G_1$ or $G_2$
forms an outerplane graph with the property that the 
external face of $U$ is a cycle with no repeated vertices.
\end{enumerate}
\end{lemma}

\begin{proof}
First, we compute the dual $D$ of $G$, where each face in (the primal graph) $G$ is a node in $D$
and there is an arc between two nodes in $D$ if their corresponding primal faces share an
edge in common.
We ignore the external face in this step.
For clarity we shall refer to vertices and edges in the primal and nodes and arcs in the dual; see Fig.~\ref{fig:split}(a).
We further augment the dual by adding an arc between two nodes in $D$ if they also share a vertex in common.
Call this augmented dual graph $D^*$.

Let the source node $s$ be the node corresponding to the edge $e_l$ and
the sink node $t$ be the node corresponding to the edge $e_r$.
We then perform a breadth-first shortest-path traversal from $s$ to
$t$ on $D^*$; see Fig.~\ref{fig:split}(b). Let $p^*$ be a shortest (augmented) path in $D^*$ obtained by this search.
We now create a (regular) path $p$ by expanding the augmented arcs added.
That is, if there is an arc $(u,v) \in p^*$ such that $u$ and $v$ share a common
vertex in $G$ but {\em not} a common edge in $G$, i.e. they are 
part of a fan around the common vertex,
we add back the regular arcs from $u$ to $v$ adjacent to this common
vertex.
The choice of going clockwise or counter-clockwise around the common
vertex depends on the previous visited arc; see
Fig.~\ref{fig:split}(c).

All of the steps described above can be easily implemented in linear
time. The details of the proof can be found in~\cite{XXX}.
\qed
\end{proof}

\begin{figure}[t]
\begin{center}\begin{tabular}{cc}
\includegraphics[width=5cm]{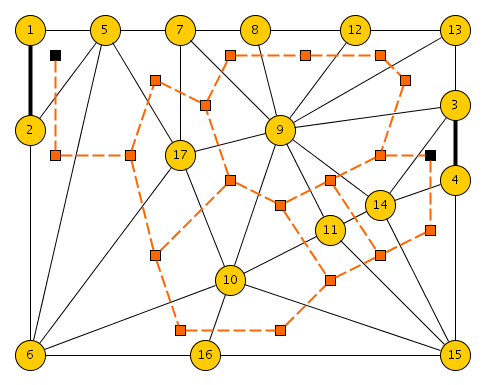} &
\includegraphics[width=5cm]{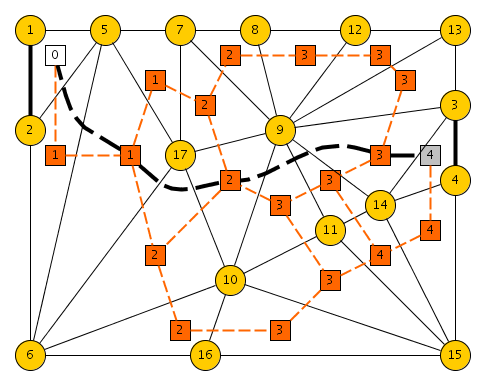}\\
(a) & (b)\\
\includegraphics[width=5cm]{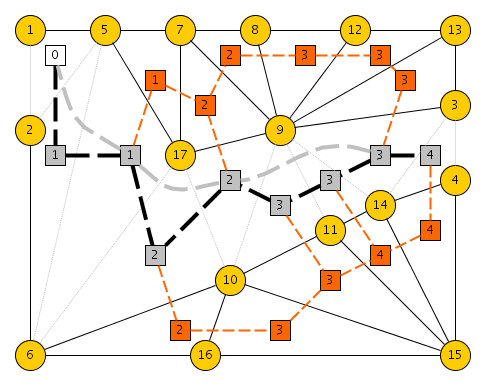} &
\includegraphics[width=5cm]{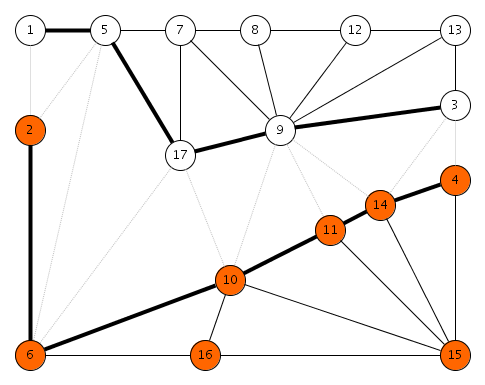}\\
(c) & (d)
\end{tabular}\vspace{-0.5cm}\end{center}
\caption{\small\sf (a) A graph $G$ and its dual $D$. 
The dark edges/nodes represent the sink and source nodes.
(b) Each dual node is labeled with its distance (in $D^*$) from the start node $0$.
A shortest path $p^*$ is drawn with thick dark arcs.  This path includes
the augmented arcs of $D^*$.
(c) The path $p$ formed after expanding the augmented arcs.
The edges from the primal that are cut by this path are shown faded.
(d) The two sets $V_1$ (light vertices) and $V_2$ (darker vertices) formed by the removal
of path $p$.
The external face of $U$ is defined by the thick edges along with the edges $(1,2)$ and $(3,4)$.}
\label{fig:split}
\end{figure}

Figure~\ref{fig:split}(d) illustrates the result of one such partition.
In some cases we might have to start and end with a set of edges
rather than just the two edges $e_l$ and $e_r$. The following
extension of Lemma~\ref{lemma:splitA} addresses this issue; the
details of the proof can be found in~\cite{XXX}.

\begin{lemma}
\label{lemma:splitB}
Given an embedded plane graph $G$ that is fully triangulated, except for the external face,
and given two vertex-disjoint chord-free paths $L$ and $R$ on that external face,
it is possible in linear time to partition $V(G)$ into two subsets $V_1$ and $V_2$ such that
\begin{enumerate}
\item the subgraphs of $G$ induced by $V_1$ and $V_2$, called $G_1$ and $G_2$, are both
connected subgraphs;
\item there exists exactly one vertex $v \in V(L)$ (say $v \in V_1$) with neighbors 
in $V(L) \setminus V_2$ (the opposite vertex set that are not part of $V(L)$),
the same holds for $V(R)$; and
\item the union $U$ of the set of faces in $G$ that are not in $G_1$ or $G_2$
forms an outerplane graph with the property that the 
external face of $U$ is a cycle with no repeated vertices.
\end{enumerate}
\end{lemma}

We can now discuss the steps for the grid drawing of the genus-1 graph $G^*$ 
with an external face formed by $\cP$.
Using Lemma~\ref{lemma:splitB}, with $L=P_4$ and $R=P_2$, divide $G^*$ into two subgraphs
$G_1$ and $G_2$.  We proceed to embed $G_1$ with $G_2$ being symmetric.
Assume without loss of generality that $G_1$ contains the bottom path, $P_3$.
Compute a canonical order of $G_1$ so that the vertices of $P_3$ are the last
vertices removed.
Place all of the vertices of $P_3$ on a
horizontal line, $p_{3,k_3}, p_{3,k_3-1}, \dots, p_{3,1}$ placed consecutively on $y=0$.
This is possible since there are no edges between them (because the path is chord-free).
Recall that the standard dPP algorithm~\cite{fpp-hdpgg-90} maintains the invariant that at the start
of each iteration, the current external face
consists of the original horizontal line and a set of
line segments of slope $\pm 1$ between consecutive vertices. 
The algorithm also maintains a ``shifting set'' for each vertex.
We modify this condition by requiring that the vertices 
on the right and left
boundary that are part of $P_2$ and $P_4$ be aligned vertically
and that the current external face might have horizontal slopes corresponding
to vertices from $P_3$; see Fig.~\ref{fig:embed}(a).
Upon insertion of a new vertex $v$, the vertex will have consecutive
neighboring vertices on the external face.
We label the left and rightmost neighbors $x_\ell$ and $x_r$.
To achieve our modified invariant, 
we insert a vertex $v$ into the current drawing depending on its
type, $0$, $1$, or $2$, as follows:

\noindent{\bf Type 0:} Vertices not belonging to a path in $\cP$ are inserted as 
with the traditional dPP algorithm.
This insertion might require up to two horizontal shifts determined by
the shifting sets; see Fig.~\ref{fig:embed}(a).

\noindent {\bf Type 1:} Vertices belonging to $P_2$, which must be placed vertically along
the right boundary, are inserted with a line 
segment of slope $+1$ between $x_\ell$ and $v$ and 
a vertical line segment between $v$ and $x_r$.
Notice that $x_r$ must also be in $P_2$.
And because $P_2$ is chord-free $x_r$ is the topmost vertex on the right side of the current
external face.
That is, $v$ can see $x_r$.
By Lemma~\ref{lemma:splitB} and the fact that the graph was fully triangulated, 
we also know that $v$ must have a vertex $x_\ell$.
This insertion requires only 1 shift, for the visibility of $x_\ell$ and $v$.
Again the remaining vertices $x_{\ell+1}, \dots, x_{r-1}$ are
connected as usual; see Fig.~\ref{fig:embed}(b).

\noindent{\bf Type 2:} Vertices belonging to $P_4$, which must be placed vertically along
the left boundary, are handled similarly to Type 1.
Because
of Lemma~\ref{lemma:splitB}, after processing both $G_1$ and $G_2$, 
 we can proceed to stitch the two portions together.
Shift the left wall of the narrower graph sufficiently
to match the width of the other graph.
For simplicity, refer to the vertices on the external face of each subgraph
that are not exclusively part of the wall or bottom row as \define{upper external vertices}.
For each subgraph, 
consider the point $p$ located at the intersection of the lines of slope $\pm 1$
extending from the left and rightmost external vertices.
Flip $G_2$ vertically placing it so that its point $p$ lies
either on or just above (in case of non-integer intersection) $G_1$'s point.
Because the edges between the upper external vertices have
slope $|m| \leq 1$ and because of the vertical separation of the two
subgraphs, every upper external vertex on $G_1$ can directly see every
upper external vertex on $G_2$.
By Lemma~\ref{lemma:splitB}, we know that the set of edges removed in the separation
along with the edges connecting the upper external vertices forms an outerplanar graph.
Therefore, we can reconnect the removed edges, joining the two subgraphs, without introducing
any crossings.

\begin{figure}[t]
\begin{center}
\begin{tabular}{cc}
\includegraphics[width=4.7cm]{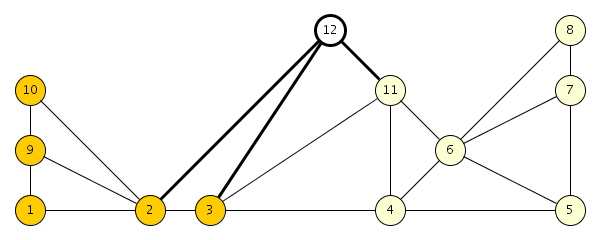} &
\includegraphics[width=4.7cm]{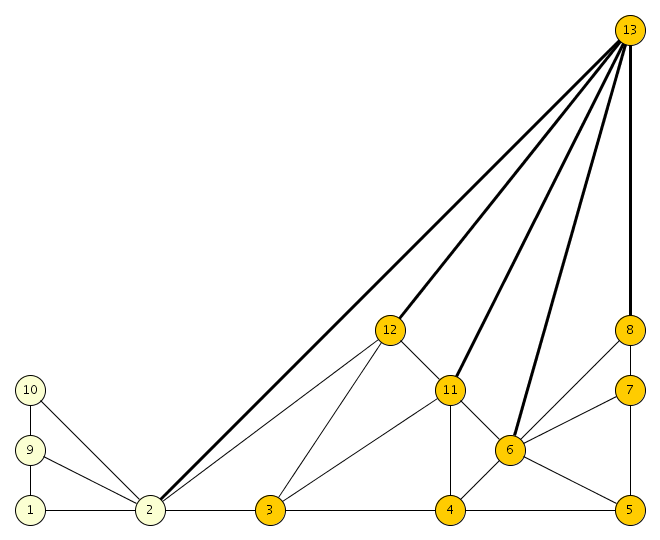}\\
(a) & (b)
\end{tabular}\end{center}\vspace{-0.5cm}
\caption{\small\sf 
(a)  The embedding process after insertion of the first 11 vertices and
the subsequent insertion of 
a Type 0 vertex with $v=12$, $x_\ell=2$, and $x_r=11$.
Note the invariant condition allowing
the two partial vertical walls $\{8,7,5\} \subset P_2$
and $\{1,9,10\} \subset P_4$.
The light vertices to the right of $12$ including $x_r$ 
have been shifted over one unit.
(b) The result of inserting a Type 1 vertex with $v=13$, $x_\ell=2$, and $x_r=8$.
Note, the light vertices to the left of and including $x_\ell=2$ are shifted
over one unit.
}
\label{fig:embed}
\end{figure}

We claim that the area of this grid is $O(n) \times O(n^2)$.
First, let us analyze the width.
From our discussion, we have accounted for each insertion step using shifts.
Since the maximum amount of shifting of 2 units is done with Type $0$ vertices,
we know that each of the two subgraphs has width at most $2n$.
In addition, the stitching stage only required a shifting of the smaller width 
subgraph.
Therefore, the width of our drawing is at most $2n$.
Ideally, the height of our drawing would also match this bound.
The stitching stage for example only adds at most $W \leq 2n$ units to the final
height.
After the insertion of each wall vertex we know that the height increases
by at most $W$.
Therefore, we know that the height is at most $Wn$ or $2n^2$
and consequently we have a correct drawing using a grid of size $O(n) \times O(n^2)$.
\qed
\end{proof}

Thus, we get a planar polygonal-schema drawing of a torodial graph $G$ in 
a rectangle, which simultaneously achieves polynomial area, a
straight frame, and straight-line edges.

\subsection{The Peel-and-Bend Algorithm}
\label{sec:embedHigher}

The case for $g>1$ is similar, but involves a few alterations.
First, we use $n=|V|$ unlike the previous sections which used
$n = |V| + g$.
The main difference, however, is that we cannot embed the outer face using only horizontal and vertical walls
unless the fundamental cycles are 
chord-free or 
unless edge bends are allowed.
Since we desire a straight-frame rendering of the polygonal
schema $\cal P$ in a
rectangle, we must allow some edge bends in this case.
The following theorem describes our resulting drawing method, which
we call the \emph{Peel-and-Bend} algorithm.

\begin{theorem}
\label{theorem:gEmbed}
Let $G^*$ be an embedded planar graph and
$\cP=\{P_1, P_2, \dots, P_{4g}\}$ in $G^*$ be a collection of $4g$ paths such that
each path $P_i=\{p_{i,1}, p_{i,2}, \dots, p_{i,k_i}\}$ is chord-free,
the last vertex of each path matches the
first vertex of the next path,
and when treated as a single cycle, $\cP$ forms the external face of $G^*$.
Let $k=\sum_{i=1}^{4g} (k_i-1)$ be the number of vertices on the external cycle.
We can draw $G^*$ on an $O(n) \times O(n^2)$ grid 
with straight-line
edges and no crossings and at most $k-3$ single-bend edges
in such a way that for each path $P_i$ on the external face the
vertices on that path form a straight line.
\end{theorem}

\begin{proof}
First, let us assume that the entire external face, represented by $\cP$, 
is completely chord-free.
That is, if two vertices on the external cycle share an edge then they are adjacent 
on the cycle.
In this case we can create a new set of 4 paths, 
$\cP'=\{P_1, \cup_{i=2, \dots, 2g}P_i, P_{2g+1}, \cup_{i=2g+2, 4g}P_i\}.$
We can then use Theorem~\ref{theorem:gOneEmbed} to prove our claim using no bends.

If, however, there exist chords on the external face, embedding the graph with 
straight-lines becomes problematic, and in fact impossible to do 
using a rectangular outer face.
By introducing a temporary bend vertex for each chord and
retriangulating the two neighboring faces, we can make the external face chord free.
Clearly this addition can be done in linear time.
Since there are at most $k$ vertices on the external face and since the graph
is planar, there are no more than $k-3$ such bend points to add.
We then proceed as before using Theorem~\ref{theorem:gOneEmbed},
subsequently replacing inserted vertices with a bend point.
\qed
\end{proof}

\section{Algorithms for Non-Toroidal Graphs}
In this section, we describe two more algorithms for producing a planar
polygonal-schema drawing of a non-toroidal graph $G$, which is given
together with its combinatorial embedding in a genus-$g$ surface,
$\cal S$, where $g>1$.
As mentioned above,
these algorithms 
provide alternative trade-offs with respect to the three primary
aesthetic criteria we desire for polygonal-schema drawings.
For the sake of space, we describe these algorithms at a very high level
and leave their details and full analysis to the final version of
this paper.

\paragraph{The Peel-and-Stretch Algorithm.}
In the Peel-and-Stretch Algorithm, we find a chord-free polygonal schema 
$\cal P$ for $G$ and cut $G$ along these edges to form a planar graph
$G^*$.
We then layout the sides of $\cal P$ in a straight-frame manner as a
regular convex polygon, with the vertices along each boundary edge
spaced as evenly as possible.
We then fix this as the outer face of $G^*$ and apply Tutte's
algorithm~\cite{t-crg-60,t-hdg-63} 
to construct a straight-line drawing of the rest
of $G^*$.
This
algorithm therefore achieves a drawing with straight-line 
edges in a (regular) straight frame, but it may require exponential
area when normalized to an integer grid, since Tutte's drawing algorithm
may generate vertices with coordinates that require $\Theta(n\log n)$
bits to represent.

\paragraph{The Peel-and-Place Algorithm.}
In the Peel-and-Place Algorithm, we start by finding a 
polygonal schema 
$\cal P$ for $G$ and cut $G$ along these edges to form a planar graph
$G^*$, as in all our algorithms.
In this case, we then create a new triangular face, $T$, and place
$G^*$ in the interior of $T$, and we fully triangulate this graph.
We then apply the dPP algorithm~\cite{fpp-hdpgg-90}
to construct a drawing of this graph in an $O(n)\times O(n)$ integer
grid with straight-line edges. 
Finally, we remove all extra edges to produce a polygonal
schema drawing of $G$.
The result will be a polygonal-schema drawing with straight-line
edges having polynomial area, but there is no guarantee that it is a
straight-frame drawing, since the dPP algorithm makes no
collinear guarantees for vertices adjacent to the vertices on the
bounding triangle.

\section{Conclusion and Future Work}
In this paper,
we present several algorithms for polygonal-schema drawings of
higher-genus graphs. 
Our method for toroidal graphs achieves drawings that 
simultaneously use straight-line edges in a straight frame and
polynomial area.
Previous algorithms for the torus were restricted
to special cases or did not always produce
polygonal-schema renderings~\cite{e-tbtdo-09,kns-dgt-01,Vodopivec20081847}.
Our methods for non-toroidal graphs can achieve any two of these
three criteria. It is an open problem
whether it is possible to achieve all three of these aesthetic
criteria for non-toroidal graphs.
To the best of our knowledge,
previous algorithms for general graphs in genus-$g$ surfaces were
restricted to those with ``nice'' polygonal
schemas~\cite{z-dgs-94}.

\bibliographystyle{abbrv}
{
\begin{small}
\vspace{-.3cm}\bibliography{stephen}
\end{small}
}

\end{document}